%% file: infodicho.tex
\documentclass[conference]{IEEEtran}

\usepackage[dvips]{graphicx}
\usepackage{mdwlist}
\usepackage{amsmath,amssymb,amsthm,dsfont,url,braket,ctable,hyperref,commath}
\usepackage[percent]{overpic}
\usepackage{mathtools,cite}
\usepackage{enumitem}

\newcommand{\Tr}{\operatorname{Tr}}

\DeclareMathOperator{\openone}{\mathds{1}}
\DeclareMathOperator*{\argmax}{arg\,max}

 \newtheorem{lmm}{Lemma}
\newtheorem{thm}{Proposition} \newtheorem{cor}{Corollary}
\newtheorem{conj}{Conjecture}
\newtheorem{algo}{Algorithm}

\begin{document}

\title{On Shor's conjecture on the accessible information of
  quantum dichotomies}

\author{
  \IEEEauthorblockN{Khac Duc An Thai}
  \IEEEauthorblockA{
    \textit{Department  of Computer Science and Engineering, Toyohashi University of Technology, Toyohashi, Japan}
  }
  \IEEEauthorblockN{Michele Dall'Arno}
  \IEEEauthorblockA{
    \textit{Department  of Computer Science and Engineering, Toyohashi University of Technology, Toyohashi, Japan}\\
    michele.dallarno.mv@tut.jp
  }
}  

\date{\today}

\maketitle

\begin{abstract}
  Around  the  turn  of  the century,  Shor  formulated  his
  well-known  and  still-open  conjecture stating  that  the
  accessible information  of any quantum dichotomy,  that is
  the maximum  amount of  classical information that  can be
  decoded from a  binary quantum encoding, is  attained by a
  von Neumann  measurement.  A  quarter of a  century later,
  new  developments   on  the   Lorenz  curves   of  quantum
  dichotomies  in  the  field of  quantum  majorization  and
  statistical comparison may provide  the key to unlock such
  a longstanding  open problem.  Here, we  first investigate
  the tradeoff relations  between accessible information and
  guessing probability  in the binary case,  thus disproving
  the  claimed monotonicity  of the  former quantity  in the
  latter that, if true, would have settled Shor's problem in
  the  qubit  case.   Our  second result  is  to  provide  a
  state-dependent generalization of  extremality for quantum
  measurements, to  characterize state-dependent extremality
  for  qubit  dichotomies,  and  to apply  such  results  to
  tighten previous results on  the accessible information of
  qubit dichotomies.
\end{abstract}

\section{Introduction}

Given a certain  encoding in the form  of classical signals,
the Shannon  mutual information quantifies the  maximum rate
at  which reliable  decoding can  be attained  with such  an
encoding.   The   quantum  generalization  of   a  classical
encoding is given by a quantum state; analogously, a quantum
measurement generalizes a  classical detector.  Accordingly,
the quantum generalization of the Shannon mutual information
is the  accessible information~\cite{Hol73a,  Hol73b, JRW94,
  FC94, Lev95,  Fuc96, Sho00,  Rus02, Kei08,  Kei09, DDJB14,
  JM15} of  the quantum  ensemble of states.  The accessible
information equals the maximum  over quantum measurements of
the Shannon mutual information  between the classical random
variable  labeling  the  states  of the  ensemble,  and  the
outcome  of  the  measurement.   As  such,  the  problem  of
computing the  accessible information of any  given ensemble
is non-convex, and analytical solutions are known only for a
limited number of symmetric ensembles.

In  1995, Levitin  showed~\cite{Lev95}  that the  accessible
information of  any dichotomic  ensemble (or  dichotomy, for
short) of  \textit{pure} states is attained  by its Helstrom
measurement, that is, by  the measurement that maximizes the
success probability  in the  discrimination of  such states.
This  result  established  a   bridge  between  the  Shannon
information-theory (the accessible information) and Bayesian
statistics  (the  guessing  probability~\cite{Bla53,  Tor70,
  AU80, Tor91,  Mat10a, Mat10b, RKW11, MOA11,  Jen12, Bus12,
  Mat14a, Mat14b,  Ren16, Jen16,  BG17, Dal19,  BST19, WW19,
  DBS20, DB24}),  exclusively in  the case  of \textit{pure}
states.  The  fact that  the accessible information  and the
guessing  probability  for   \textit{mixed}  states  are  in
general \textit{not}  attained by  the same  measurement was
observed~\cite{Fuc96}  shortly  afterwards  by Fuchs  in  an
extensive study of the accessible information.

Since two pure  states span a qubit, and  since the Helstrom
measurement of  any ensemble  is a von  Neumann measurement,
Levitin went  on conjecturing that  for any ensemble  of $d$
states in  a $d$-dimensional Hilbert space,  the measurement
attaining the accessible information  would be a von Neumann
measurement.   Shortly   thereafter,  this   conjecture  was
disproved~\cite{Sho00}  by  Shor,  who however  proposed  an
alternative conjecture based on  numerical evidence by Fuchs
and Perez:  that for  any dichotomy in  arbitrary dimension,
the  accessible  information  would  be attained  by  a  von
Neumann  measurement.  In this  work,  we  focus on  such  a
conjecture.

A few  years later,  Keil~\cite{Kei08, Kei09}  proved Shor's
conjecture  for  any,   generally  mixed,  qubit  dichotomy.
Keil's proof is not constructive;  that is, the optimal (von
Neumann) measurement  for any given qubit  dichotomy remains
given  as   the  implicit   solution  of   a  transcendental
equation. Hence, Keil concluded  his analysis by proposing a
conjecture himself, that here we rephrase by saying that the
accessible  information  problem  would  be  a  quasi-convex
problem.   If  true,  this  statement  would  of  course  be
remarkable, as it would make it possible to apply well-known
polynomial-time algorithms, whose convergence is guaranteed,
to the computation of the accessible information.

Not long  afterwards, the claim was  first made~\cite{Paw10}
and  then reiterated~\cite{PB11,  HG12, Paw12},  in contrast
with the aforementioned  results by Fuchs~\cite{Fuc96}, that
for  any  binary (both  in  the  input  and in  the  output)
probability  distribution,  the Shannon  mutual  information
would be a monotone function  of the guessing probability, a
fact that, coupled with Keil's  result, would imply that the
accessible  information  and  the guessing  probability  are
attained  by   the  same   measurement  for   any,  possibly
\textit{mixed}, qubit dichotomy.  While the specific results
of  Refs.~\cite{Paw10, PB11,  Paw12} do  not depend  on this
observation,  the  observation  \textit{per se}  is  clearly
incorrect (see again Ref.~\cite{Fuc96}).

Our first  contribution is  to compute the  tradeoff between
the guessing probability and  the mutual information, and to
derive   necessary  and   sufficient   conditions  for   the
monotonicity of the latter quantity  in the former.  We show
that, already  in the  binary case, monotonicity  holds only
for   a    zero-measure   subset   of    joint   probability
distributions, and we  derive a closed-form characterization
of such a  subset. Hence, Keil's conjecture  remains to date
an open problem.

A new development  to Keil's conjecture came in  the form of
the closed-form characterization~\cite{Dal19}  of the Lorenz
curve of any given qubit dichotomy. For any given dichotomy,
its  Lorenz curve  is the  boundary of  its testing  region,
which  in turn  is  given by  the set  of  ordered pairs  of
probabilities attainable  by the dichotomy over  any quantum
measurement elements, also known  as effects.  Lorenz curves
and  testing   regions  play  a  crucial   role  in  quantum
majorization   and   statistical  comparison,   specifically
through   results   such   as   the   celebrated   Blackwell
theorem~\cite{Bla53} and its  quantum counterpart by Alberti
and  Uhlmann~\cite{AU80}, where  the  task  is to  establish
whether a quantum channel exists that maps a given dichotomy
into another one.

Our second  result is to  leverage known facts  from quantum
majorization  to  introduce  a  state-dependent  concept  of
\textit{extremality}, that generalizes  the usual concept of
extremality~\cite{Par99,DLP05}  for   quantum  measurements.
Specifically,  given  a  set  of states,  a  measurement  is
extremal with respect to such a set whenever it generates an
extremal conditional probability distribution upon the input
of  such states.   In  terms of  statistical comparison  and
quantum majorization,  in the binary case  such a definition
corresponds  to the  set of  effects that  generate extremal
points  of  the  testing region~\cite{AU80,  Tor91,  Mat10a,
  Mat10b, RKW11, MOA11, Jen12, Bus12, Mat14a, Mat14b, Ren16,
  Jen16,  BG17, Dal19,  BST19,  WW19, DBS20,  DB24} (on  the
Lorenz curve) of  the family of states.   We characterize in
closed  form such  a  set  of effects  for  any given  qubit
dichotomy;   in   particular,    we   show   that,   perhaps
surprisingly,  such  a  set   happens  to  be  symmetrically
centered around the Helstrom  measurement that maximizes the
guessing probability  for the  balanced case  (uniform prior
over the two states).  Notice that extremality is in general
a necessary, but not sufficient,  condition for this form of
state-dependent extremality.

In  the  context  of  the accessible  information  of  qubit
dichotomies, the optimization can then be restricted without
loss of  generality only  to those von  Neumann measurements
that generate extremal conditional probability distributions
over  the given  dichotomy, thus  tightening Keil's  result.
Moreover, we can strengthen Keil's conjecture as a statement
on   the   pseudo-convexity   of  such   a   quantity   over
state-dependently extremal measurements.

The     paper    is     structured    as     follows.     In
Section~\ref{sect:classical} we  discuss the problem  of the
information-guessing    tradeoff    relation   for    binary
probability         distributions,          while         in
Section~\ref{sect:quantum} we study its quantum counterpart.
We    conclude    by     summarizing    our    results    in
Section~\ref{sect:conclusion}.

\section{Tradeoff between mutual information and guessing probability}
\label{sect:classical}

Given two finite  and discrete random variables  $X$ and $Y$
with $x  \in \mathcal{X}$  and $y  \in \mathcal{Y}$,  in the
following   we   denote  with   $\mathcal{P}_{|\mathcal{X}|,
  |\mathcal{Y}|}$  the   set  of  their   joint  probability
distributions.        We      denote       with      $\delta
\mathcal{P}_{|\mathcal{X}|, |\mathcal{Y}|}$  the boundary of
$\mathcal{P}_{|\mathcal{X}|,  |\mathcal{Y}|}$,  that is  the
set  of joint  probability distributions  with at  least one
zero   entry.     We   denote   with    $\overline{   \delta
  \mathcal{P}}_{|\mathcal{X}|,  |\mathcal{Y}|}$ the  bulk of
$\mathcal{P}_{|\mathcal{X}|,  |\mathcal{Y}|}$,  that is  the
complement     of    $\delta     \mathcal{P}_{|\mathcal{X}|,
  |\mathcal{Y}|}$.   Notice  that,   except  for  degenerate
cases,    $\delta    \mathcal{P}$   and    $\overline{\delta
  \mathcal{P}}_{|\mathcal{X}|, |\mathcal{Y}|}$ have zero and
unit measure, respectively.

The mutual information~\cite{CT06}  $I(X:Y)$ between $X$ and
$Y$ is given by
\begin{align}
  \label{eq:minfo}
  I(X:Y) := H(X) - H(X|Y),
\end{align}
where $H(X) := -\sum_x p_{X=x}  \log p_{X=x}$ is the entropy
of  $X$  and $H(X|Y)  :=  \sum_y  p_{Y=y} H(X|Y=y)$  is  the
conditional  entropy   of  $X$   given  $Y$.    The  maximum
probability $P_{X|Y}$  of correctly  guessing $X$  given $Y$
(usually referred to as guessing probability~\cite{CT06}) is
given by
\begin{align}
  \label{eq:guess}
  P_{X|Y} := \sum_y p_{Y=y} P_{X|Y=y},
\end{align}
where  $P_{X|Y=y}  :=  \max_x p_{X=x|Y=y}$  is  the  maximum
probability of correctly guessing $X$ given that $Y = y$.

We  derive  (Proposition~\ref{thm:tradeoff})  the  trade-off
relation  between the  mutual information  $I(X:Z)$ and  the
guessing probability  $P_{X|Z}$, when the marginal  $P_X$ is
given, that is
\begin{align}
  \label{eq:tradeoff}
  \max_{\substack{p_{X,Z} \\  P_{X|Z} \le P}} I(X:Z),
\end{align}
for any given  marginal $P_X$ and any $P  \ge 0$.  Moreover,
we derive (Proposition~\ref{thm:characterization}) necessary
and  sufficient  conditions  for   any  given  binary  joint
probability   distribution   $p_{X,Y}$    to   satisfy   the
implication
\begin{align}
  \label{eq:false}
  P_{X|Y} \ge P_{X|Z} \ \Longrightarrow \ I(X:Y) \ge I(X:Z),
\end{align}
for any binary joint probability distribution $p_{X,Z}$.

Notice  that the  guessing  probability  $P_{X|Y}$ is  lower
bounded by
\begin{align}
  \label{eq:minguess}
  P_{X|Y} \ge \max p_X,
\end{align}
with equality  corresponding to  trivial guessing  the input
with     largest     prior      probability.      Due     to
Eq.~\eqref{eq:minguess}, a  trivial necessary  condition for
the  constraint  in   Eq.~\eqref{eq:tradeoff}  and  for  the
hypothesis   in  Eq.~\eqref{eq:false}   to  hold   for  some
$p_{X,Z}$ is that $P_{X|Y} > \max p_X$.

Before proceeding, let us review well-known related results.
Fano's inequality~\cite{CT06} provides an upper bound on the
conditional  entropy  $H(X|Y)$  in  terms  of  the  guessing
probability $P_{X|Y}$, as follows
\begin{align}
  \label{eq:fano}
  h(P_{X|Y})  + \left(1  - P_{X|Y}\right)  \log|\mathcal{X}|
  \ge H(X|Y),
\end{align}
where $h(p)$ denotes the  binary entropy of probability $p$,
that is $h(p) := -p \log p - (1-p) \log(1-p)$.  Notice that,
in the case $|\mathcal{X}|=|\mathcal{Y}|$, it is possible to
strengthen the above  bound by replacing $\log|\mathcal{X}|$
with $\log(|\mathcal{X}|-1)$.

Analogously,   the    Hellman-Raviv   inequality~\cite{HR70}
provides a  lower bound on the  conditional entropy $H(X|Y)$
in terms of the guessing probability $P_{X|Y}$ as follows
\begin{align}
  \label{eq:hellman-raviv}
  H(X|Y) \ge 2 \left(1 - P_{X|Y} \right).
\end{align}

Let  us  consider  now  three  finite  and  discrete  random
variables    $X$,   $Y$,    and   $Z$.     By   substituting
Eqs.~\eqref{eq:fano}    and~\eqref{eq:hellman-raviv}    into
Eq.~\eqref{eq:minfo}, one immediately  obtains the following
relations
\begin{align}
  \label{eq:sufficient}
  & 2 \left( 1 - P_{X|Z} \right) \ge h(P_{X|Y}) + \left( 1 -
  P_{X|Y}                                            \right)
  \log|\mathcal{X}|\nonumber\\ \Longrightarrow\ & I(X:Y) \ge
  I(X:Z), 
\end{align}
and
\begin{align}
  \label{eq:necessary}
  &  I(X:Y)   \ge  I(X:Z)  \nonumber\\   \Longrightarrow\  &
  h(P_{X|Z}) + \left( 1  - P_{X|Z} \right) \log|\mathcal{X}|
  \ge 2 \left( 1 - P_{X|Y} \right).
\end{align}
Notice       again       that,       in       the       case
$|\mathcal{X}|=|\mathcal{Y}|$, it is  possible to strengthen
the  above  bounds  by  replacing  $\log|\mathcal{X}|$  with
$\log(|\mathcal{X}|-1)$.

Notice   that   Eq.~\eqref{eq:false}   first   appeared   in
Refs.~\cite{Paw10, PB11}, where it  has been claimed to hold
under the  sole assumption  that $|\mathcal{X}| =  2$, where
$|\mathcal{X}|$ denotes  the cardinality  of the  support of
$X$.  Ref.~\cite{HG12} showed that Eq.~\eqref{eq:false} does
not hold already for $|\mathcal{X}| = |\mathcal{Y}| = 2$ and
$|\mathcal{Z}|  =   4$.   However,   both  Refs.~\cite{HG12}
and~\cite{Paw12} still claimed  that Eq.~\eqref{eq:false} at
least holds  in the completely binary  case $|\mathcal{X}| =
|\mathcal{Y}|=     |\mathcal{Z}|     =     2$.      Finally,
Eq.~\eqref{eq:sufficient} also appears in Ref.~\cite{Paw12},
but the proof therein only  covers the case $| \mathcal{X} |
= 2$.

\subsection{Binary Joint Probability distributions}

In the following we focus  on the binary case $|\mathcal{X}|
=  |\mathcal{Y}| =  | \mathcal{Z}  |  = 2$.   We derive  the
trade-off relation in Eq.~\eqref{eq:tradeoff} and a complete
characterization of the distributions $p_{X,Y}$ that satisfy
Eq.~\eqref{eq:false} for any $p_{X,Z}$. From our results, it
follows  that  for  almost every  binary  joint  probability
distribution  $p_{X,Y}$  for  which  Eq.~\eqref{eq:minguess}
holds with  inequality, except a zero-measure  subset of the
boundary $\delta  \mathcal{P}_{2,2}$, there exists  a binary
joint  probability  distribution   $p_{X,Z}$  that  violates
Eq.~\eqref{eq:false}.  Moreover, our  proof is constructive,
that   is    a   distribution   $p_{X,Z}$    that   violates
Eq.~\eqref{eq:false} is given in  closed form as an explicit
function of $p_{X, Y}$ only.

We need  to introduce a convenient  parameterization for any
given binary  joint probability distribution. Let  us define
the following matrices
\begin{align*}
  M_0  := \begin{pmatrix}  1 &  1  \\ 1  & 1  \end{pmatrix},
  \qquad   M_1  :=   \begin{pmatrix}   1  &   1   \\  -1   &
    -1 \end{pmatrix},\\ \qquad M_2 := \begin{pmatrix} 1 & -1
    \\ 1 & -1 \end{pmatrix}, \qquad M_3 := \begin{pmatrix} 1
    & -1 \\ -1 & 1 \end{pmatrix}.
\end{align*}
Notice that  matrices $\{ M_k/2 \}_{k=0}^3$  are orthonormal
with  respect  to  the   Hilbert-Schmidt  product,  that  is
$\Tr[M_k^T M_j]/4  = \delta_{k,j}$.  Any given  binary joint
probability distribution $p_{X,Y}$ can be written as
\begin{align*}
  p_{X,Y}  = \begin{pmatrix}
  p_{X=0,Y=0} &  p_{X=0,Y=1} \\  p_{X=1,Y=0} &  p_{X=1,Y=1}
  \end{pmatrix} =   \frac14 \sum_{k=0}^3 c_k M_k
\end{align*}
where coefficients $\{ c_k \}_{k=0}^3$ can be readily found by
\begin{align*}
  c_k = \Tr[p_{X,Y} M_k^T].
\end{align*}
Hence,  one has  $c_0  = 1$. Then, we  set $c  :=  (1, a,  b,
\lambda)$ and make the identification
\begin{align}
  \label{eq:joint}
  p(a, b, \lambda) := p_{X,Y}.
\end{align}

With  this   parameterization,  by  direct   computation  the
marginals of $P_{X,Y}$ are given by
\begin{align}
  \label{eq:marginalx}
  p_X \left(  a \right) & =  \left( \frac{1+a}2, \frac{1-a}2
  \right),\\
  \label{eq:marginaly}
  p_Y \left( b \right) & = \left( \frac{1+b}2, \frac{1-b}2 \right),
\end{align}
from  which  one immediately  has  $a,  b \in  [-1,1]$.   By
explicit computation, we have:
\begin{align*}
  p_{X, Y} = \frac14 \sum_{k=0}^3 c_k M_k
  = \frac14
  \begin{pmatrix}
    1 + a + b + \lambda & 1 + a - b - \lambda \\ 1 - a + b -
    \lambda & 1 - a - b + \lambda
  \end{pmatrix}
\end{align*}
Hence,  the non-negativity  of  each entry  of $p_{X,Y}$  is
equivalent to the condition
\begin{align}
  \label{eq:domainlambda}
  \lambda \in \left[ -1 +  \left| a+b \right|, 1 - \left|a-b
    \right|\right].
\end{align}

Notice  that  the  conditions  $-1  \le  a,  b  \le  1$  and
Eq.~\eqref{eq:domainlambda} are equivalent to the conditions
$-1 \le a, \lambda \le 1$ and
\begin{align}
  \label{eq:domainb}
  b \in \left[ -1 + \left| a + \lambda \right|, 1 - \left| a
    - \lambda \right| \right].
\end{align}

\subsection{Properties of Guessing Probability and Mutual Information} \label{sec:properties}

Now we  discuss some properties of  the guessing probability
$P_{X|Y}$ and  of the mutual information  $I(X:Y)$ that will
be useful in the following.

\begin{lmm}
  The guessing probability $P_{X|Y}$ is given by
  \begin{align}
    \label{eq:guessbin}
    P_{X|Y}  =  \frac{1+\max\left(  \left| a  \right|,  \left|
      \lambda \right| \right)}2 =: P(a, \lambda).
  \end{align}
\end{lmm}

\begin{proof}
  By the definition given by Eq. \eqref{eq:guess}, one has:
  \begin{align*}
    P_{X|Y} &  = \sum_y  p_{Y=y} P_{X|Y=y} =  \sum_y p_{Y=y}
    \max_x  p_{X=x|Y=y}   \\  &  =  \sum_y   p_{Y=y}  \max_x
    \frac{p_{X, Y}}{p_{Y=y}} = \sum_y \max_x p_{X, Y} \\ & =
    \frac{1+b}{2} \max_x p_{X=x|Y=0}  + \frac{1-b}{2} \max_x
    p_{X=x|Y=1}.
  \end{align*}
  Using Bayes' Theorem, one has
  \begin{align*}
    p_{X|Y} = \frac{p_{X, Y}}{p_Y} = \frac14
    \begin{pmatrix}
      \frac{1 + a + b  + \lambda}{\frac{1+b}{2}} & \frac{1 +
        a - b - \lambda}{\frac{1-b}{2}} \\ \frac{1 - a + b -
        \lambda}{\frac{1+b}{2}}  &   \frac{1  -  a  -   b  +
        \lambda}{\frac{1-b}{2}}
    \end{pmatrix}.
  \end{align*}
  Hence,
  \begin{align*}
    P_{X|Y} & = \frac14 (\max{(1 + a  + b + \lambda, 1 - a +
      b - \lambda)} \\ & + \max{(1 +  a - b - \lambda, 1 - a
      - b + \lambda)} ) \\ & = \frac{2 + \abs{a + \lambda} +
      \abs{a - \lambda}}{4}.
  \end{align*}
  To complete the proof, we need to prove that:
  \begin{align*}
    |a + \lambda| + |a - \lambda| = 2 \max(|a|, |\lambda|).
  \end{align*}
  Since  $a$ and  $\lambda$  play dual  roles, consider  two
  cases:
  \subsection*{Case 1: $a$ and $\lambda$ have the same sign}
  \begin{itemize}
  \item  Assume  $|a|  \ge  |\lambda|$, with  $a  >  0$  and
    $\lambda > 0$:
    \begin{align*}
      |a +  \lambda| +  |a - \lambda|  = a +  \lambda +  a -
      \lambda = 2a = 2\max(|a|, |\lambda|).
    \end{align*}
    Similarly,  if $a  < 0$  and $\lambda  < 0$,  the result
    holds by symmetry.
  \item If  $|a| < |\lambda|$,  we reverse the roles  of $a$
    and $\lambda$, leading to the same conclusion.
  \end{itemize}
  \subsection*{Case 2: $a$ and $\lambda$ have different signs}
  \begin{itemize}
  \item  Assume  $|a| \geq  |\lambda|$,  with  $a >  0$  and
    $\lambda < 0$:
    \begin{align*}
      |a +  \lambda| +  |a - \lambda|  = a +  \lambda +  a -
      \lambda = 2a = 2\max(|a|, |\lambda|).
    \end{align*}
    Similarly, for  $a <  0$ and $\lambda  > 0$,  the result
    holds by symmetry.
  \item If  $|a| < |\lambda|$,  we reverse the roles  of $a$
    and $\lambda$, leading to the same conclusion.
  \end{itemize}
  By evaluating all cases, we conclude that
  \[
  |a + \lambda| + |a - \lambda| = 2 \max(|a|, |\lambda|),
  \]
  which completes the proof.
\end{proof}

Notice   that   the   guessing  probability   $P_{X|Y}$   is
independent of parameter $b$, is  even in $a$ and $\lambda$,
that is
\begin{align}
  \label{eq:guesssym}
  P(a, \lambda)  = P(-a, \lambda)  = P(a, -\lambda)  = P(-a,
  -\lambda).
\end{align}

The mutual information $I(X:Y)$  is invariant upon permuting
rows  or  columns  of  $p_{X,Y}$, which  correspond  to  the
transformations $(a, b, \lambda)  \to (-a, b, -\lambda)$ and
$(a, b, \lambda) \to (a, -b, -\lambda)$, hence
\begin{align}
  \label{eq:infosym}
  I(a, b, \lambda) = I(-a, b, -\lambda) = I(a, -b, -\lambda)
  = I(-a, -b, \lambda).
\end{align}

Due  to Eq.~\eqref{eq:guesssym}  and Eq.~\eqref{eq:infosym},
w.l.o.g. we can restrict to $a \ge 0$ and $\lambda \ge 0$.

Starting  from  the  well-known   convexity  of  the  mutual
information  as a  function of  the conditional  probability
distribution and motivated  by the fact that  the maximum of
convex  functions  is  attained  on the  boundary  of  their
domain,   in   the   following   lemma  we   arrive   to   a
characterization of  the analytical behavior of  the mutual
information on its boundary.

\begin{lmm}
  \label{lmm:info}
  The mutual information $I(X:Y)$ as a function of $\lambda$
  has the following properties for any $a \ge 0$ and $b$:
  \begin{enumerate}
  \item  \label{item:convlambda}   $I(X:Y)$  is   convex  in
    $\lambda$;
  \item \label{item:minlambda}  $I(X:Y)$ attains  its unique
    global minimum in $\lambda = ab$, and $I(a, b, ab) = 0$;
  \end{enumerate}

  The mutual information  $I(X:Y)$ as a function  of $b$ has
  the following  properties for any  $a \ge 0$  and $\lambda
  \ge 0$:
  \begin{enumerate}[resume]
  \item \label{item:convb} $I(X:Y)$ is convex in $b$;
  \item \label{item:minb} $I(X:Y)$ attains its unique global
    minimum in $b = \lambda/a$ if  $\lambda \le a$ and in $b
    = a/\lambda$ otherwise;
  \item \label{item:maxb} $I(X:Y)$ attains its unique global
    maximum on  $\delta \mathcal{P}_{2,2}$,  specifically in
    $b = \lambda + a - 1$.
  \end{enumerate}

  The mutual information $I(X:Y) |_{b = \lambda + a - 1}$ as
  a function  of $\lambda$ has the  following properties for
  any $a \ge 0$:
  \begin{enumerate}[resume]
  \item \label{item:convlambda2} $I(X:Y) |_{b  = \lambda + a
    - 1}$ is convex in $\lambda$.
  \end{enumerate}
\end{lmm}

\begin{proof}
  Remember  that $I$  is a  smooth  function of  all of  its
  variables.
  
  Property~\ref{item:convlambda} immediately follows from
  \begin{align*}
    \frac{\partial^2 I}{\partial \lambda^2}  = \frac1{16} \sum_{x, y}
      \frac1{p_{X = x, Y = y}} \ge 0.
  \end{align*}

  Property~\ref{item:minlambda}  immediately follows  from the
  fact that
  \begin{align*}
    \frac{\partial  I}{\partial  \lambda}   =  \frac14  \log
    \frac{p_{X = 0, Y  = 0} \; p_{X = 1, Y = 1}}{p_{X  = 0, Y =
        1} \; p_{X = 1, Y = 0}} = 0
  \end{align*}
  if and only  if $\lambda = a b$, and  from the observation
  that $\lambda = ab$ satisfies Eq.~\eqref{eq:domainlambda}.
  
  Property~\ref{item:convb} follows from the fact that
  \begin{align*}
    & \frac{\partial^2 I}{\partial b^2}  \\ = & \frac1{b^2 -
      1}  + \frac1{16}  \sum_{x,  y} \frac1{p_{X  =  x, Y  =
        y}}\\  =   &  \frac1{b^2  -  1}   +  \frac12  \left(
    \frac{\left(1  +  b   \right)}{\left(1  +  b\right)^2  -
      \left(a+\lambda\right)^2}   +    \frac{\left(1   -   b
      \right)}{\left(1         -        b\right)^2         -
      \left(a-\lambda\right)^2} \right) \\  \ge & \frac1{b^2
      - 1} +  \frac12 \left( \frac1{1  + b} + \frac1{1  - b}
    \right) = 0.
  \end{align*}
  
  Property~\ref{item:minb} immediately follow  from the fact
  that
  \begin{align*}
    \frac{\partial I}{\partial b}  = \frac14 \log \frac{p_{X
        = 0, Y = 0} \; p_{X = 1, Y = 0} \; p_{Y = 1}}{p_{X =
        0, Y = 1} \; p_{X = 1, Y = 1} \; p_{Y = 0}} = 0
  \end{align*}
  if and only if $b = \lambda/a$ or $b = a / \lambda$.

  To prove property~\ref{item:maxb},  we need to distinguish
  two cases.

  If $\lambda \le a$ one has
  \begin{align*}
    \Delta \left(  a, \lambda \right) :=  I \left( a, -  a +
    \lambda + 1, \lambda \right) - I \left( a, a + \lambda -
    1, \lambda \right).
  \end{align*}
  By  explicit   computation,  one   has  $\Delta(a,   0)  =
  0$. Moreover, one has
  \begin{align*}
    \frac{\partial \Delta}{\partial \lambda}  = \frac12 \log
    \frac{1 - \lambda^2}{\left( 2 - a \right)^2 - \lambda^2}
    \le 0.
  \end{align*}
  Hence,  for  any  $\lambda  \ge 0$  one  has  $\Delta  (a,
  \lambda) \le 0$.
  
  If $\lambda > a$ one has
  \begin{align*}
    \Delta  \left( a,  \lambda \right)  := I  \left( a,  a -
    \lambda + 1, \lambda \right) - I \left( a, a + \lambda -
    1, \lambda \right).
  \end{align*}
  By  explicit computation,  one has  $\Delta(0, \lambda)  =
  0$. Moreover, one has
  \begin{align*}
    \frac{\partial  \Delta}{\partial   a}  =   \frac12  \log
    \frac{1 - a^2}{\left(  2 - \lambda \right)^2  - a^2} \le
    0.
  \end{align*}
  Hence, for any $a \ge 0$  one has $\Delta (a, \lambda) \le
  0$.

  Property~\ref{item:convlambda2} immediately follows from
  \begin{align*}
  \frac{\partial^2 I(X:Y) |_{b  = \lambda + a - 1}}{\partial 
  \lambda^2} = \frac{1 - a}{2 \left( 1 - \lambda \right) 
  \left( 2 - a - \lambda \right)} \ge  0.
  \end{align*}
\end{proof}

The  results  of   Lemma~\ref{lmm:info}  are  summarized  in
Fig.~\ref{fig:info}.
\begin{figure}[h!]
  \centering
  \input{fig01}
  \caption{Mutual   information    of   binary   probability
    distribution $p_{X,Y}$ on the  $\lambda, b$ plane, for a
    certain  value of  $a$.   The  rectangle represents  the
    feasible region  satisfying the  constraints on  $b$ and
    $\lambda$.}
  \label{fig:info}
\end{figure}
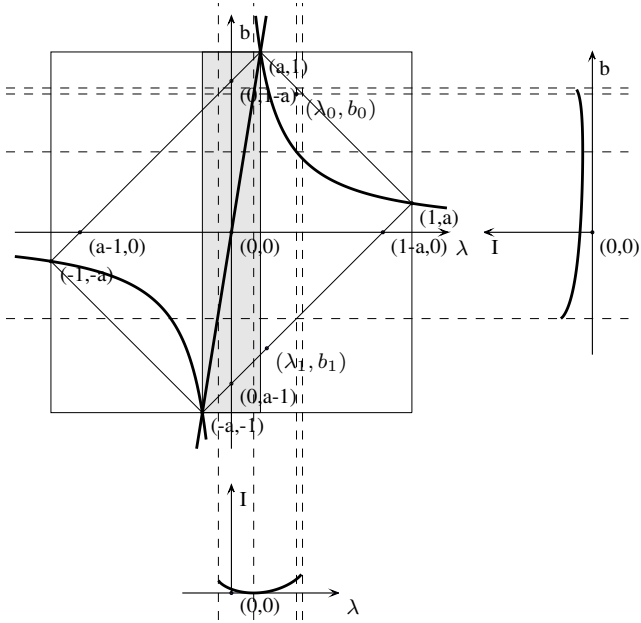

\subsection{Tradeoff relations between mutual information and guessing probability}

Our first main result is  to derive the tradeoff relation in
Eq.~\eqref{eq:tradeoff} between the  mutual information $I(X
: Y)$ and the guessing probability $P_{X|Y}$, over any binary
joint probability distribution $p_{X,  Y}$ and for any given
marginal probability $p_X$.

\begin{thm}[Information-guessing tradeoff]
  \label{thm:tradeoff}
  For  any given  binary  marginal probability  distribution
  $p_X$ (w.l.o.g.  take  $p_{X = 0} \ge p_{X =  1}$) and any
  upper bound $P$ on the guessing probability $P_{X|Y}$, the
  binary  joint  probability   distribution  $p_{X,Y}$  that
  maximizes the mutual information $I(X : Y)$ is given by
  \begin{align}
    \label{eq:tradeoff}
    \argmax_{\substack{p_{X,  Y}\\\sum_y  p_{X,  Y  =  y}  =
        p_X\\ P_{X|Y} \le  P}} I(X : Y) =  p \left(a^*, b^*,
    \lambda^* \right).
  \end{align}
  where $a^* = 2  p_{X = 0} - 1$, $b^* =  \lambda^* + (a^* -
  1)$, $\lambda^* = 2 P - 1$.
\end{thm}

\begin{proof}
  One has
  \begin{align*}
    &  \max_{\substack{p_{X,  Y}\\\sum_y  p_{X,  Y  =  y}  =
        p_X\\   P_{X|Y}  \le   P}}   I(X  :   Y)   \\  =   &
    \max_{\substack{\lambda\\\lambda   \le   2P-1}}   \left(
    \max_b  I \left(a^*,  b, \lambda  \right)\right) \\  = &
    \max_{\substack{\lambda\\\lambda \le 2P-1}} I \left(a^*,
    \lambda + a^*  - 1, \lambda \right) \\ =  & I \left(a^*,
    b^*, \lambda^* \right).
  \end{align*}
  Where the first equality follows  from the facts that, due
  to Eq.~\eqref{eq:marginalx}, the  constraint $\sum_y p_{X,
    Y = y} = p_X$ immediately gives  $a^* = 2 p_{X = 0} - 1$
  and,  due   to  Eq.~\eqref{eq:guessbin},   the  constraint
  $P_{X|Y} \le  P$ is equivalent  to $a^* \le  2 P -  1$ and
  $\lambda \le  2 P -  1$; the second equality  follows from
  Property~\ref{item:maxb}   of  Lemma~\ref{lmm:info};   the
  third           equality            follows           from
  Property~\ref{item:convlambda2}                        and
  Property~\ref{item:minlambda} of Lemma~\ref{lmm:info}.
\end{proof}

Our second main result is the characterization of the set of
binary joint  probability distributions $p_{X,Y}$  such that
Eq.~\eqref{eq:false} holds for  any binary joint probability
distribution $p_{X,Z}$.

\begin{thm}[Non-monotonicity]
  \label{thm:characterization}
  A   binary   joint  probability   distribution   $p_{X,Y}$
  satisfies  Eq.~\eqref{eq:false}   for  any   binary  joint
  probability   distribution  $p_{X,Z}$   if  and   only  if
  \begin{align*}
    \Tr p_{X, Y} \ge p_{X = 0},
  \end{align*}
  and
  \begin{align*}
    P_{X|Y} = p_{Y = 0} + \left(1 - p_{X = 0} \right).
  \end{align*}
\end{thm}

\begin{proof}
  We need to solve the following problem
  \begin{align*}
    \max_{\substack{p_{X,Z}\\\sum_z  p_{X  =  0, Z  =  z}  =
        \sum_y p_{X  = 0,  Y =  y}\\P_{X|Z} \le  P_{X|Y}}} I
    \left( X : Z \right) \le I(X : Y).
  \end{align*}

  Due to Proposition~\ref{thm:tradeoff},  the maximum in the
  l.h.s. is  attained uniquely by $p(a^*,  b^*, \lambda^*)$,
  with  $a^* =  2 \sum_y  p_{X =  0,  Y =  y} -  1$, $b^*  =
  \lambda^* + a^* - 1$, and  $\lambda^* = 2 P_{X|Y} - 1$. By
  setting $p_{X, Y}  = p(a^*, b^*, \lambda^*)$,  we have the
  following.

  The first  condition is  tautologically satisfied  (in the
  end,  it is  just the  requirement of  consistency between
  $p_{X, Y}$ and $p_{X, Z}$).

  From the  third condition we have  $\Tr p_{X, Y} =  P_{X |
    Y}$, which  in turn is  equivalent to $\Tr p_{X,  Y} \ge
  p_{X = 0}$.

  Finally, from  the second  condition we  have $p_{X,  Y} =
  p_{Y = 0} + (1 - p_{X = 0})$, which proves the statement.  
\end{proof}

By Proposition~\ref{thm:characterization}, the set of binary
joint   probability   distributions  $p_{X,Y}$   such   that
Eq.~\eqref{eq:false} holds for  any binary joint probability
distribution $p_{X,Z}$  is a  strict subset of  the boundary
$\delta  \mathcal{P}$ of  the  set  $\mathcal{P}$ of  binary
joint  probability   distributions,  and   hence  it   is  a
zero-measure   set,  in   contrast  with   claims  made   in
Refs.~\cite{Paw10, PB11,  HG12, Paw12} (that however  do not
affect the validity of the results therein).

\subsection{Sufficient conditions for the validity of Eq.~\eqref{eq:false}}

Propositions~\ref{thm:tradeoff}
and~\ref{thm:characterization} show that imposing conditions
on $p_{X,Y}$ only is in  general not enough to guarantee the
validity of  Eq.~\eqref{eq:false}. Hence, it is  relevant to
derive   sufficient   conditions   for   the   validity   of
Eq.~\eqref{eq:false}  in  terms  of  $p_{X,Y}$  \textit{and}
$p_{X,Z}$.   In  particular,  the  properties  discussed  in
Section~\ref{sec:properties} make the study of the case $p_Y
= p_Z$ analytically tractable, so we will restrict to such a
case.

First, we consider  the case where, along with  $p_Y = p_Z$,
we additionally require the uniformity of $p_X$, and we have
the  following  sufficient  condition for  the  validity  of
Eq.~\eqref{eq:false}.

\begin{cor}[Monotonicity]
  \label{cor:uniformx}
  Equation~\eqref{eq:false} holds  for any two  binary joint
  probability  distributions  $p_{X,Y}$ and  $p_{X,Z}$  such
  that $p_X$ is uniform and $p_Y = p_Z$. Moreover
  \begin{align*}
    P_{X|Y} = P_{X|Z} \Longleftrightarrow I(X:Y) = I(X:Z).
    \end{align*}
\end{cor}

\begin{proof}
  Under the assumptions of the  corollary one has $p_{X,Y} =
  p(0, b, \lambda)$  and $p_{X,Z} =: p(0, b,  \mu)$.  Due to
  Eqs.~\eqref{eq:guessbin}    and~\eqref{eq:infosym},    the
  guessing  probability $P_{X|Y}  = P(0,  \lambda)$ and  the
  mutual information  $I(X:Y) =  I(0, b, \lambda)$  are even
  functions of $\lambda$.   Moreover, $P_{X|Y}$ and $I(X:Y)$
  are  increasing   functions  in  $|\lambda|$,   hence  the
  statement is proved.
\end{proof}

Second, we consider the case  where, along with $p_Y = p_Z$,
we additionally  require the  uniformity of $p_Y$  and $p_Z$
and  that $\Tr  p_{X,Y}  \ge p_{X  = 0}$,  and  we have  the
following   sufficient  condition   for   the  validity   of
Eq.~\eqref{eq:false}.

\begin{cor}
  \label{cor:uniformyz}
  Equation~\eqref{eq:false} holds  for any two  binary joint
  probability  distributions  $p_{X,Y}$ and  $p_{X,Z}$  such
  that $p_Y$ and $p_Z$ are uniform and that $\Tr p_{X,Y} \ge
  p_{X = 0}$. Moreover
  \begin{align*}
    P_{X|Y}  =  P_{X|Z}  \ \Longleftrightarrow  \  I(X:Y)  =
    I(X:Z).
  \end{align*}
\end{cor}

\begin{proof}
  Under the assumptions of the  corollary one has $p_{X,Y} =
  p(a,  0, \lambda)$  and  $p_{X,Z} =:  p(a,  0, \mu)$  with
  $\lambda   \ge  a$.    The  condition   on  the   guessing
  probabilities   $P_{X|Y} \ge P_{X|Z}$  implies   also
  $\lambda \ge \mu$. Hence the statement follows.
\end{proof}

Along  with  the  restriction  $p_Y =  p_Z$,  to  prove  the
validity of  Eq.~\eqref{eq:false} for any $p_{X,Z}$  we have
additionally   assumed   the    uniformity   of   $p_X$   in
Corollary~\ref{cor:uniformx} or the  uniformity of $p_Y$ and
$p_Z$   in  Corollary~\ref{cor:uniformyz}.    The  following
counterexample shows that the restriction  $p_Y = p_Z$ is by
itself insufficient for the validity of Eq.~\eqref{eq:false}
for any $p_{X,Z}$.

Consider   the  following   two  binary   joint  probability
distributions $p_{X,Y}$ and $p_{X,Z}$
\begin{align*}
  p_{X,Y} = \begin{pmatrix}
    0.48 & 0.11 \\
    0.11 & 0.30
  \end{pmatrix}, \qquad
   p_{X,Z} = \begin{pmatrix}
    0.21 & 0.38 \\
    0.38 & 0.03
  \end{pmatrix},
\end{align*}
for which $p_Y = p_Z$. By direct inspection one has
\begin{align*}
  I(X : Y)  < 0.23, \qquad &      I(X : Z) > 0.26,
\end{align*}
where logarithms are taken to base $2$, whereas
\begin{align*}
  P_{X|Y} = 0.78, \qquad P_{X|Z} = 0.76,
\end{align*}
which disproves the validity of Eq.~\eqref{eq:false}.

\section{Measurements attaining the extremal points of Lorenz curves} 
\label{sect:quantum}

\subsection{An explicit parameterization for Lorenz curves}

Extremal POVMs  are relevant  because, for instance,  in the
maximization of  convex functions  one can  restrict without
loss of  generality to extremal  POVMs. The set  of extremal
POVMs has been characterized~\cite{Par99, DLP05}.

Such a  concept of extremality of  a POVM within the  set of
POVMs is state-independent. However,  if a particular family
of  states is  given and  the task  if to  optimize a  given
payoff function  over POVMs, one  does not need  to consider
all extremal  POVMs, but only  the subset of  extremal POVMs
that generate extremal conditional probability distributions
over such a family of states, according to the Born rule. We
have, therefore, a concept of  extremality for POVMs that is
state-dependent. Such POVMs generate  extremal points of the
testing region, whose boundary is given by the Lorenz curve,
of the given family of states.

Let us  consider the simplest  non-trivial case, that  is, a
family  of  two states  (a  dichotomy)  $(\rho, \sigma)$  is
given, and  the task  is to find  the (of  course, extremal)
two-outcomes  POVMs   that  generate   extremal  conditional
probability  distributions over  such  a  dichotomy. Such  a
problem is  equivalent to finding the  (of course, extremal)
effects that generate extremal  points in the testing region
of such a dichotomy.

Let us assume,  without loss of generality,  that $\rho \neq
\sigma$, otherwise  the problem trivializes. To  begin with,
we need to derive an explicit expression for an operator
\begin{align}
  \label{eq:omega}
  \omega =  \mu \rho + \left(1  - \mu\right) \sigma
\end{align}
given by the affine combination of $\rho$ and $\sigma$, that
is  orthogonal, in  the  Hilbert-Schmidt sense,  to $\rho  -
\sigma$.  Solving the linear equation
\begin{align*}
  \Tr \left[ \left(  \mu \rho + \left(1  - \mu\right) \sigma
    \right) \left( \rho - \sigma \right) \right] = 0,
\end{align*}
immediately gives
\begin{align}
  \label{eq:mu}
  \mu  = \frac{\Tr  \sigma^2 -  \Tr \rho  \sigma}{\Tr \left(
    \rho - \sigma \right)^2},
\end{align}
which is well-defined due to $\rho \neq \sigma$.

Notice   that   the   linear   independence   of   operators
$\tilde\omega := \omega - \openone/d$ and $\rho - \sigma$ is
equivalent to the linear independence of $\tilde\rho := \rho
- \openone/d$  and  $\tilde\sigma  := \sigma  -  \openone/d$
($\openone$  and $d$  denote the  identity operator  and the
dimension  of  the  Hilbert space,  respectively).   Indeed,
$\tilde\omega$ cannot be  the null vector because  it can be
written  as  the  affine  combination  (that  is,  a  linear
combination  whose coefficients  can  not be  simultaneously
null) of $\tilde\rho$ and $\tilde\sigma$ as follows
\begin{align*}
  \tilde\omega =  \mu \tilde\rho  + \left(  1 -  \mu \right)
  \tilde\sigma.
\end{align*}
Hence,  the  orthogonality  of $\tilde\omega$  and  $\rho  -
\sigma$   (also   not    null)   guarantees   their   linear
independence.      Conversely,    if     $\tilde\rho$    and
$\tilde\sigma$   are    linearly   dependent,    then   also
$\tilde\omega$ and  $\rho -  \sigma$ are  linearly dependent
since the latter are,  by definition, linear combinations of
the former.

\begin{lmm}
  For  any given  dichotomy $\rho  \neq \sigma$,  the affine
  combination $\omega$ as given by Eqs.~\eqref{eq:omega} and
  ~\eqref{eq:mu} has purity given by
  \begin{align}
    \label{eq:omega_purity}
    \Tr \omega^2  = \frac{\Tr  \rho^2 \Tr \sigma^2  - \left(
      \Tr \rho  \sigma \right)^2}{\Tr  \left( \rho  - \sigma
      \right)^2} < 1.
  \end{align}
\end{lmm}

\begin{proof}
  From    the    definition    of    $\omega$    given    by
  Eq. \eqref{eq:omega}, by replacing the definition of $\mu$
  given by Eq.~\eqref{eq:mu}, one has
  \begin{align*}
    \omega  = &  \frac{\Tr \sigma^2  - \Tr  \rho \sigma}{\Tr
      \left(  \rho  - \sigma  \right)^2}  \rho  + \left(1  -
    \frac{\Tr \sigma^2 - \Tr  \rho \sigma}{\Tr \left( \rho -
      \sigma  \right)^2} \right)  \sigma\\ =  & \frac{\left(
      \Tr \sigma^2 -  \Tr \rho \sigma \right)  \rho + \left(
      \Tr  \rho^2 -  \Tr  \rho \sigma  \right) \sigma}  {\Tr
      \left( \rho - \sigma \right)^2} .
  \end{align*}
  The  purity of  the  affine combination  $\omega$ is  then
  given by
  \begin{align*}
    \Tr \omega^2 & = \Tr  \left[ \frac{\left( \Tr \sigma^2 -
        \Tr \rho \sigma  \right) \rho + \left(  \Tr \rho^2 -
        \Tr \rho  \sigma \right) \sigma} {\Tr  \left( \rho -
        \sigma \right)^2} \right]^2 \\ & = \frac{ \left( \Tr
      \sigma^2  - \Tr  \rho  \sigma \right)^2  \Tr \rho^2  +
      \left(  \Tr \rho^2  -  \Tr \rho  \sigma \right)^2  \Tr
      \sigma^2} {\left  [\Tr \left( \rho -  \sigma \right)^2
        \right]^2} \\ & + \frac{2  \left( \Tr \sigma^2 - \Tr
      \rho  \sigma  \right) \left(  \Tr  \rho^2  - \Tr  \rho
      \sigma  \right) \Tr  \rho \sigma}  {\left[ \Tr  \left(
        \rho - \sigma \right)^2 \right]^2} . \\
  \end{align*}
  The numerator of the equation above simplifies to
  \begin{align*}
    & \left[ \left( \Tr \sigma^2  \right)^2 - 2 \Tr \sigma^2
      \Tr  \rho \sigma  + \left(  \Tr \rho  \sigma \right)^2
      \right] \Tr  \rho^2 \\  & +  \left[ \left(  \Tr \rho^2
      \right)^2 - 2 \Tr \rho^2  \Tr \rho \sigma + \left( \Tr
      \rho \sigma  \right)^2 \right] \Tr  \sigma^2 \\ &  + 2
    \left[ \Tr \sigma^2  \Tr \rho^2 - \left(  \Tr \sigma^2 +
      \Tr \rho^2 \right)  \Tr \rho \sigma +  \left( \Tr \rho
      \sigma \right)^2  \right] \Tr \rho  \sigma \\ =  & \Tr
    \rho^2 \Tr  \sigma^2 \left(  \Tr \sigma^2  - 2  \Tr \rho
    \sigma  + \Tr  \rho^2 \right)  \\  & -  \left( \Tr  \rho
    \sigma \right)^2 \left( \Tr \sigma^2 - 2 \Tr \rho \sigma
    +  \Tr \rho^2  \right)  \\  = &  \left[  \Tr \rho^2  \Tr
      \sigma^2 -  \left( \Tr  \rho \sigma  \right)^2 \right]
    \Tr \left( \rho - \sigma \right)^2.
  \end{align*}
  hence the equality in Eq.~\eqref{eq:omega_purity}.

  To prove the inequality in Eq.~\eqref{eq:omega_purity}, we
  multiplying both  its sides by  $\Tr \left( \rho  - \sigma
  \right)^2$, thus obtaining
  \begin{align*}
    \Tr  \rho^2  \Tr  \sigma^2  -  \left(  \Tr  \rho  \sigma
    \right)^2 < \Tr \left( \rho - \sigma \right)^2,
  \end{align*}
  or equivalently
  \begin{align*}
    \Tr \rho^2 \Tr \sigma^2 - \Tr  \rho^2 - \Tr \sigma^2 + 1
    < \left( 1 - \Tr \rho \sigma \right)^2.
  \end{align*}
  One has  
  \begin{align*}
    0 \leq 1  - \sqrt{\Tr \rho^2} \sqrt{\Tr \sigma^2}  < 1 -
    \Tr \rho \sigma.
  \end{align*}
  where the first inequality  holds because $\Tr \rho^2, \Tr
  \sigma^2 \leq 1$, and the second (strict) inequality holds
  because the Cauchy-Schwarz holds  strictly for (non equal)
  states, that is
  \begin{align*}
    \left(  \Tr  \rho  \sigma  \right)^2 <  \Tr  \rho^2  \Tr
    \sigma^2.
  \end{align*}

  By squaring both sides, this inequality is equivalent to
  \begin{align*}
    \left(  1   -  \sqrt{\Tr  \rho^2}   \sqrt{\Tr  \sigma^2}
    \right)^2 \leq \left( 1 - \Tr \rho \sigma \right)^2.
  \end{align*}
  We can easily prove that
  \begin{align*}
    \Tr \rho^2 \Tr \sigma^2 - \Tr  \rho^2 - \Tr \sigma^2 + 1
    \leq  \left(1 -  \sqrt{\Tr  \rho^2} \sqrt{\Tr  \sigma^2}
    \right)^2,
  \end{align*}
  or equivalently
  \begin{align*}
    \Tr  \rho^2  +  \Tr   \sigma^2  -  2  \sqrt{\Tr  \rho^2}
    \sqrt{\Tr \sigma^2} \geq 0.
  \end{align*}
  This holds true because 
  \begin{align*}
    \left(\sqrt{\Tr \rho^2} - \sqrt{\Tr \sigma^2} \right)^2 \geq 0.
  \end{align*}
  Hence, 
  \begin{align*}
    \Tr \rho^2 \Tr \sigma^2 - \Tr  \rho^2 - \Tr \sigma^2 + 1
    < \left( 1 - \Tr \rho \sigma \right)^2
  \end{align*}
  which completes the proof.
\end{proof}

Hence, we  can introduce  a convenient  parameterization for
Helstrom's matrices as follows
\begin{align}
  \label{eq:helstrom}
  H  \left( \lambda  \right)  = \lambda  \left(  \mu \rho  +
  \left(  1 -  \mu \right)  \sigma \right)  - \left(  \rho -
  \sigma \right),
\end{align}
where $\lambda \in  \mathbb{R}$. Two-outcomes POVMs generate
an   extremal  conditional   probability  distribution   for
dichotomy  $(\rho, \sigma)$  if  and only  if their  effects
correspond to  the projectors  on the positive  and negative
parts of $H ( \lambda )$, for some $\lambda$.

The  minimum  of  the  mutual information  is  attained  for
$\lambda =  \pm\infty$ in Eq.~\eqref{eq:helstrom},  since in
this  case  and  only  in  this case,  one  has  the  mutual
independence of the input and output random variables.

\subsection{Measurements  attaining the  extremal points  of
  the testing region of a qubit dichotomy}
  
Notice  that   for  any   given  qubit   dichotomy  ($\rho$,
$\sigma$),  the  affine  combination $\omega$  as  given  by
Eq.~\eqref{eq:omega}  and  ~\eqref{eq:mu}  is  a  (non-pure)
state.

Rather than  optimizing the accessible information  over the
set of all von Neumann  measurements, it clearly suffices to
only consider measurements that  generate extremal points on
the testing  region.  The  following proposition  provides a
closed-form characterization of such measurements.

\begin{thm}
  For  any  given  qubit  dichotomy  $(\rho,  \sigma)$,  the
  maximum  over  von  Neumann  measurements  of  any  convex
  objective  function (such  as the  mutual information)  is
  attained  for  $\lambda$ in  Eq.~\eqref{eq:helstrom}  such
  that $-\lambda^* \le \lambda \le \lambda^*$, where
  \begin{align}
    \label{eq:lambda}
    \lambda^* =  \frac{\Tr \left( \rho -  \sigma \right)^2}{
      \sqrt{\Tr \left( \rho -  \sigma \right)^2 - \Tr \rho^2
        \Tr \sigma^2 + \left( \Tr \rho \sigma \right)^2}}.
  \end{align}
\end{thm}

\begin{proof}
  Measurements attaining the extremal  points of the testing
  region  are those  for  which $\lambda_-  \le \lambda  \le
  \lambda_+$ in Eq.~\eqref{eq:helstrom}, where $\lambda_\pm$
  are  the solutions  of the  quadratic equation  $\det H  (
  \lambda ) = 0$. By explicit computation one has
  \begin{align*}
    -2 \det H \left( \lambda  \right) = \Tr H \left( \lambda
    \right)^2  -   \left(  \Tr  H  \left(   \lambda  \right)
    \right)^2.
  \end{align*}
  By explicit computation one immediately has
  \begin{align*}
    \Tr{H \left( \lambda \right)} = \lambda.
  \end{align*}
  To  compute  $\Tr   H  (  \lambda  )^2$,   we  proceed  as
  follows.  From the  definition  of  $H(\lambda)$ given  by
  Eq.~\eqref{eq:helstrom},  due to  the orthogonality  (a la
  Hilbert-Schmidt) of  $\omega = \mu  \rho + \left( 1  - \mu
  \right) \sigma$ and $\rho - \sigma$, one has
  \begin{align*}
    \Tr H  \left( \lambda \right)^2  & = \Tr  \left[ \lambda
      \omega - \left( \rho - \sigma \right) \right]^2 \\ & =
    \lambda^2  \Tr  \omega^2  +  \Tr \left(  \rho  -  \sigma
    \right)^2.
  \end{align*}
  By replacing the  definition of $\Tr \omega^2$  given  by
  Eq.~\eqref{eq:omega_purity}, one has
  \begin{align*}
    \Tr H  \left( \lambda \right)^2 =  & \lambda^2 \frac{\Tr
      \rho^2  \Tr   \sigma^2  -   \left(  \Tr   \rho  \sigma
      \right)^2}{\Tr \left(  \rho - \sigma \right)^2}  + \Tr
    \left( \rho - \sigma \right)^2.
  \end{align*}
  Hence, the condition $\det H ( \lambda ) = 0$ is equivalent to
  \begin{align*}
    & \lambda^2 \left[  1 - \frac{\Tr \rho^2  \Tr \sigma^2 -
        \left( \Tr \rho \sigma  \right)^2}{\Tr \left( \rho -
        \sigma \right)^2} \right] = \Tr \left( \rho - \sigma
    \right)^2,
  \end{align*}
  or equivalently
  \begin{align*}
    \lambda^2 \frac{\Tr \left( \rho - \sigma \right)^2 - \Tr
      \rho^2  \Tr   \sigma^2  +   \left(  \Tr   \rho  \sigma
      \right)^2}{\Tr \left(  \rho - \sigma \right)^2}  = \Tr
    \left( \rho - \sigma \right)^2,
  \end{align*}
  from which
  \begin{align*}
    \lambda^2  =  \frac{\left(  \Tr  \left(  \rho  -  \sigma
      \right)^2  \right)^2}{   \Tr  \left(  \rho   -  \sigma
      \right)^2 - \Tr \rho^2 \Tr  \sigma^2 + \left( \Tr \rho
      \sigma \right)^2},
  \end{align*}
  from  which Eq.~\eqref{eq:lambda}  immediately follows  if
  the denominator  is strictly  positive.  To show  that the
  denominator   is   strictly   positive,  we   proceed   as
  follows. By explicit computation one has
  \begin{align*}
    \Tr  \left( \rho  - \sigma\right)^2  \Tr \omega^2  = \Tr
    \rho^2 \Tr \sigma^2 - \left( \Tr \rho \sigma \right)^2.
  \end{align*}
  We  know that  $\omega$ is  a non-pure  state (i.e.   $\Tr
  \omega^2 < 1$) from Eq.~\eqref{eq:omega_purity}, hence one
  has
  \begin{align*}
    \Tr  \left(  \rho -  \sigma\right)^2  >  \Tr \rho^2  \Tr
    \sigma^2 - \left( \Tr \rho \sigma \right)^2,
  \end{align*}
  which proves  the strict positivity of  the denominator in
  the   expression   for   $\lambda^2$   above   and   hence
  Eq.~\eqref{eq:lambda}.
\end{proof}

The value  $\lambda_H$ of  $\lambda$ for  which $H(\lambda)$
equals the Helstrom matrix $p_0  \rho - p_1 \sigma$, for any
probabilities $p_0, p_1 \ge 0$ ($p_0 + p_1 = 1$) is given by
\begin{align*}
  \lambda_H = \frac{p_1 - p_0}{p_0 - \mu p_0 + p_1 \mu}.
\end{align*}

To see  this, we equate,  up to  a constant factor  $k$, the
projection on the null-trace  subspace of the $H(\lambda_H)$
matrix and of the Helstrom matrix, as follows
\begin{align*}
  &  \lambda_H \left[  \mu  \left(  \rho -  \frac{\openone}2
    \right)  +  \left(1  -   \mu  \right)  \left(  \sigma  -
    \frac{\openone}2    \right)    \right]    -    \rho    +
  \frac{\openone}2  +  \sigma  - \frac{\openone}2\\  =  &  k
  \left[ p_0  \left( \rho  - \frac{\openone}2 \right)  - p_1
    \left( \sigma - \frac{\openone}2 \right) \right],
\end{align*}
from which
\begin{align*}
  & \left( \lambda_H  \mu - 1 - k p_0  \right) \left( \rho -
  \frac{\openone}2  \right)  \\  &   +  \left(  \lambda_H  -
  \lambda_H  \mu  +  1  +  k p_1  \right)  \left(  \sigma  -
  \frac{\openone}2 \right) = 0.
\end{align*}
Due to  the linear independence  of $\rho -  \openone/2$ and
$\sigma - \openone/2$ one has
\begin{align*}
  \begin{cases}
    \lambda_H \mu -  1 - k p_0 = 0,\\  \lambda_H - \lambda_H
    \mu + 1 + k p_1 = 0,
  \end{cases}
\end{align*}
from which
\begin{align*}
  k = \frac{\lambda_H \mu - 1}{p_0}.
\end{align*}
and hence
\begin{align*}
  \lambda_H  -  \mu  \lambda_H  + 1  +  \frac{p_1}{p_0}  \mu
  \lambda_H - \frac{p_1}{p_0} = 0,
\end{align*}
thus the statement.

Notice that~\cite{Lev95} $\lambda_H$  attains the accessible
information if the dichotomy  comprises pure states, for any
probabilities.

\subsection{Keil's conjecture}

To  showcase  the  convenience of  the  parameterization  in
Eq.~\eqref{eq:helstrom}, let  us show  that a  long standing
conjecture by Keil on the  accessible information of a qubit
dichotomy, once  reframed with  such a  parameterization, is
equivalent   to  the   quasi-convexity  of   the  accessible
information.   Keil  made   the  following  conjecture  (see
Conjecture~(2),  page  77  of Ref.~\cite{Kei09}),  based  on
numerical evidence:
\begin{quotation}
  For  two  states  of  a   qubit,  there  exists  only  two
  stationary  points of  the mutual  information if  the the
  number of  outcomes of  the measurements is  restricted to
  two and  both lie in the  same plane as the  states in the
  Bloch representation.  One of the stationary points is the
  global minimum and the other one is the global maximum.
\end{quotation}

A function is  quasi-concave if and only if for  any $x$ and
$y$ for any $0 \le \alpha \le 1$ one has
\begin{align*}
  f \left(  \left( 1 - \alpha  \right) x + \alpha  y \right)
  \ge \min  \left( f  \left( x \right),  f \left(  y \right)
  \right) .
\end{align*}
Hence,  Keil's  conjecture above  can  be  rewritten as  the
following  quasi-concave   formulation  of   the  accessible
information problem.
\begin{conj}
  \label{conj:quasiconcavity}
  The mutual  information of  any given qubit  dichotomy and
  any von  Neumann measurement given by  the eigenvectors of
  Eq.~\eqref{eq:helstrom} for  any $-\infty \le  \lambda \le
  +\infty$ is a quasi-concave function of $\lambda$.
\end{conj}

Due to the results of  the previous section, without loss of
generality one  can restrict Keil's conjecture  by replacing
the  infinite domain  $[-\infty, +\infty]$  with the  finite
interval $[ - \lambda_*, \lambda_*]$; the domain of validity
of  the new  conjecture (possibly  the entire  set of  qubit
dichotomies)    is   the    same   as    that   of    Keil's
conjecture.

Notice that  the restriction to $[  - \lambda_*, \lambda_*]$
eliminates the global  minima ($x = \pm  \infty$), where the
function  is  not  pseudo-concave.  Indeed,  a  function  is
pseudo-concave if and only if for any $x$ and $y$ one has
\begin{align*}
  \frac{\partial f \left( x \right)}{\partial x} \left(y - x
  \right) \le 0 \Rightarrow f  \left( y \right) \le f \left(
  x \right).
\end{align*}

This fact  allows us to  conjecture the following,  where we
replace the quasi-concavity with pseudo-concavity.
\begin{conj}
  \label{conj:keil2}
  The mutual  information of  any given qubit  dichotomy and
  any von  Neumann measurement given by  the eigenvectors of
  Eq.~\eqref{eq:helstrom} for  any $- \lambda_*  \le \lambda
  \le \lambda_*$ is a pseudo-concave function of $\lambda$.
\end{conj}
Notice  that concavity  implies  pseudo-concavity which,  in
turn, implies quasi-concavity.

In general,  the computation of the  accessible information,
even for  a qubit dichotomy,  is a non-convex  problem.  Not
surprisingly,   therefore,   known    algorithms   for   its
computation such as  SOMIM~\cite{REK05, SAE07, PFTV92, TD25}
are not guaranteed to converge  to the actual optimal value.
However,  under  conjecture~\ref{conj:keil2}, the  following
algorithm for the computation  of the accessible information
of any given qubit ensemble is guaranteed to converge to the
actual optimal value.

\begin{algo}
  \label{algo:bisect}
  Given  a   qubit  dichotomy  $\rho,  \sigma$   with  prior
  probability     distribution     $p_0,     p_1$,     under
  Conjecture~\ref{conj:keil2} the following converges to its
  accessible information:
  \begin{enumerate}
  \item   Initialize   $\lambda_{\min}   =   -   \lambda^*$,
    $\lambda_{\max} = \lambda^*$, and $\lambda' = 0$,
  \item \label{item:beginloop} Compute the derivative
    \begin{align*}
      I'  :=  \left.   \frac{\partial  I}{\partial  \lambda}
      \right|_{\lambda = \lambda'}
    \end{align*}
    of the  mutual information  of the given  dichotomy over
    the   eigenvectors   of   $H(\lambda)$   as   given   by
    Eq.~\eqref{eq:helstrom}  with  respect to  $\lambda$  in
    $\lambda = \lambda'$,
  \item if the derivative $I'$ is negative, set
    \begin{align*}
      \lambda_{\max} = \lambda',
    \end{align*}
    else
    \begin{align*}
      \lambda_{\min} = \lambda',
    \end{align*}
  \item Set
    \begin{align*}
      \lambda' = \frac{\lambda_{\min} + \lambda_{\max}}2.
    \end{align*}
  \item  Repeat  from  step~\ref{item:beginloop}  until  the
    desired  precision   is  achieved;  output   the  mutual
    information of the given dichotomy over the eigenvectors
    of $H(\lambda')$.
  \end{enumerate}
\end{algo}

This is  a bisecting algorithm, and  therefore is guaranteed
(under Conjecture~\ref{conj:keil2}) to converge to arbitrary
precision in  logarithmic time  to the  actual value  of the
accessible       information.       We       provide      an
implementation~\cite{TD25} of Algorithm~\ref{algo:bisect}.

\section{Conclusion}
\label{sect:conclusion}

We  investigated the  tradeoff relations  between accessible
information  and guessing  probability of  ensembles of  two
quantum states,  or dichotomies.   Our first result  was the
closed-form   characterization   of   the  set   of   binary
conditional probability  distributions for which  the mutual
information is a monotone  in the guessing probability, thus
disproving previous  statements on the monotonicity  of such
quantities.    Our  second   result   was  the   closed-form
characterization of the set  of all von Neumann measurements
that  generate extremal  probability distributions  when fed
any given  qubit dichotomy, thus tightening  a conjecture by
Keil   on   the   measurement   attaining   the   accessible
information.

\section{Acknowledgments}

K.~D.~A.~T.   acknowledges support  from  the Department  of
Computer  Science and  Engineering, Toyohashi  University of
Technology. M.~D.  acknowledges  support from the Department
of Computer Science and Engineering, Toyohashi University of
Technology, from the International  Research Unit of Quantum
Information,  Kyoto University,  and from  the JSPS  KAKENHI
grant number JP20K03774.

\end{document}

%% file: fig01.tex
\setlength{\unitlength}{0.85pt}
\begin{picture}(0,0)
\includegraphics[scale=0.85]{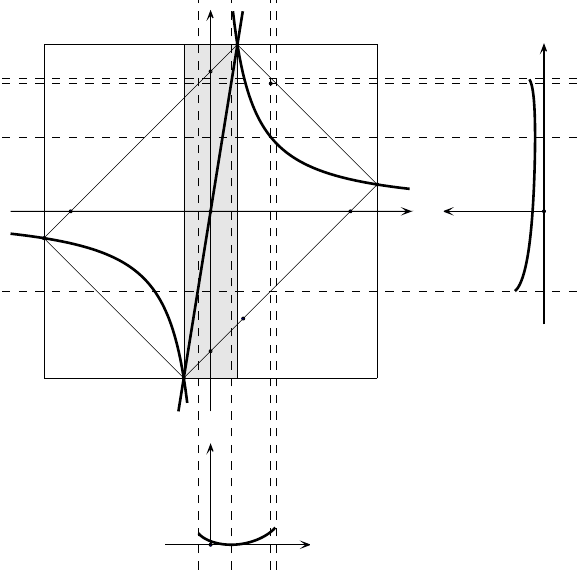}
\end{picture}
\begin{picture}(278,278)(-3,4)
\fontsize{8}{0}
\selectfont\put(194.102,174.56){\makebox(0,0)[l]{\textcolor[rgb]{0,0,0}{{$\lambda$}}}}
\fontsize{8}{0}
\selectfont\put(98.08,270.582){\makebox(0,0)[l]{\textcolor[rgb]{0,0,0}{{b}}}}
\fontsize{8}{0}
\selectfont\put(98.08,174.56){\makebox(0,0)[l]{\textcolor[rgb]{0,0,0}{{(0,0)}}}}
\fontsize{8}{0}
\selectfont\put(178.098,187.433){\makebox(0,0)[l]{\textcolor[rgb]{0,0,0}{{(1,a)}}}}
\fontsize{8}{0}
\selectfont\put(110.953,254.578){\makebox(0,0)[l]{\textcolor[rgb]{0,0,0}{{(a,1)}}}}
\fontsize{8}{0}
\selectfont\put(18.062,161.687){\makebox(0,0)[l]{\textcolor[rgb]{0,0,0}{{(-1,-a)}}}}
\fontsize{8}{0}
\selectfont\put(85.207,94.542){\makebox(0,0)[l]{\textcolor[rgb]{0,0,0}{{(-a,-1)}}}}
\fontsize{8}{0}
\selectfont\put(165.225,174.56){\makebox(0,0)[l]{\textcolor[rgb]{0,0,0}{{(1-a,0)}}}}
\fontsize{8}{0}
\selectfont\put(98.08,241.705){\makebox(0,0)[l]{\textcolor[rgb]{0,0,0}{{(0,1-a)}}}}
\fontsize{8}{0}
\selectfont\put(30.935,174.56){\makebox(0,0)[l]{\textcolor[rgb]{0,0,0}{{(a-1,0)}}}}
\fontsize{8}{0}
\selectfont\put(98.08,107.415){\makebox(0,0)[l]{\textcolor[rgb]{0,0,0}{{(0,a-1)}}}}
\fontsize{8}{0}
\selectfont\put(145.607,14.5236){\makebox(0,0)[l]{\textcolor[rgb]{0,0,0}{{$\lambda$}}}}
\fontsize{8}{0}
\selectfont\put(98.08,62.535){\makebox(0,0)[l]{\textcolor[rgb]{0,0,0}{{I}}}}
\fontsize{8}{0}
\selectfont\put(98.08,14.5236){\makebox(0,0)[l]{\textcolor[rgb]{0,0,0}{{(0,0)}}}}
\fontsize{8}{0}
\selectfont\put(210.105,174.56){\makebox(0,0)[l]{\textcolor[rgb]{0,0,0}{{I}}}}
\fontsize{8}{0}
\selectfont\put(258.116,254.597){\makebox(0,0)[l]{\textcolor[rgb]{0,0,0}{{b}}}}
\fontsize{8}{0}
\selectfont\put(258.116,174.56){\makebox(0,0)[l]{\textcolor[rgb]{0,0,0}{{(0,0)}}}}
\fontsize{8}{0}
\selectfont\put(126.938,235.929){\makebox(0,0)[l]{\textcolor[rgb]{0,0,0}{{$(\lambda_0, b_0)$}}}}
\fontsize{8}{0}
\selectfont\put(113.754,123.089){\makebox(0,0)[l]{\textcolor[rgb]{0,0,0}{{$(\lambda_1, b_1)$}}}}
\end{picture}